\documentclass[11pt]{article} 
\usepackage{listings}
\usepackage{fullpage}
\date{}
\usepackage{amsfonts,amssymb,amsmath,amsthm}

\usepackage{verbatim}

\newtheorem{proposition}{Proposition}[section]
\newtheorem{theorem}[proposition]{Theorem}

\newtheorem{corollary}[proposition]{Corollary}

\newcommand{\ident}[1]{\textit{#1}\rule{0cm}{1ex}}
\gdef\dash---{\thinspace---\hskip.16667em\relax}
\gdef\smdash--{\thinspace--\hskip.16667em\relax}
\gdef\op|{\,|\;}
\newcommand{\emn}[1]{{\em #1\/}}

\newcommand{\T}{\mathcal{T}}
\newcommand{\E}{\mathcal{E}}

\newcommand{\D}{\mathcal{D}}

\newcommand{\Naturals}{\mathbb{N}}

\newcommand{\PFD}{\ensuremath{\mathcal{P}}}

\newcommand{\MFD}{\ensuremath{\mathcal{M}}}

\newcommand{\seq}{w}
\newcommand{\conf}{\Sigma}
\newcommand{\setConf}{\Sigma^\star}
\newcommand{\initConf}{\hat{\conf}}
\newcommand{\initSetConf}{\hat{\Sigma}^\star}
\newcommand{\seqprobconf}{W(\initSetConf,\setConf)}
\newcommand{\problem}{P} %%{\mathrm{P}}
\newcommand{\Alg}{A}
\newcommand{\Trans}[1]{\ident{T}_{#1}}
\newcommand{\TrAlg}{\widetilde{A}}
\newcommand{\redCT}{\succeq^\ident{CT}}
\newcommand{\redJT}{\succeq^\ident{JT}}

\newcommand{\redSolv}{\succeq^{\ident{s}}}

\lstdefinelanguage{delta}
{
    keywords={Variables, Augmentations, Task, const, boolean, integer, array, begin, read, write, send, receive, while, if, then, else, do, true, false, procedure, requires, ensures, function, wait, until, forall, foreach, return, exit, abort, crash, define, type, new, initially, set},
    sensitive=true,
    literate={:=}{{$\gets~$}}{1} {<=}{{$\leq$}}{1} {>=}{{$\geq$}}{1} {/=}{{$\neq$}}{1},
    morecomment=[s]{/*}{*/},
    commentstyle=\itshape
}

\lstdefinestyle{nonumbers}{numbers=none}
\lstdefinestyle{numbers}{numbers=left, numberstyle=\normalsize, stepnumber=1, numbersep=7pt, xleftmargin=10pt}

\begin{document}

\title{Solvability-Based Comparison of Failure Detectors}

\author{
\begin{tabular}{c c}
Srikanth Sastry & Josef Widder\\
%CSAIL& Dept. of Computer Science and Engr.\\
Google, Inc. & TU Wien\\
\end{tabular}
}

%\date{\today}

\maketitle 

\begin{abstract}   

Failure detectors are oracles that have been introduced to provide
     processes in asynchronous systems with information about faults.
This information can then be used to solve problems otherwise unsolvable in asynchronous systems.
A natural question is on the ``minimum amount of information'' a
     failure detector has to provide for a given problem.
This question is classically addressed using a relation that states
     that a failure detector~$\D$ is stronger (that is, provides
     ``more, or better, information'') than a failure detector $\D'$ if $\D$ can
     be used to implement~$\D'$.
It has recently been shown that this classic implementability relation
     has some drawbacks.
To overcome this, different relations have been defined, one of which
     states that a failure detector $\D$ is stronger than $\D'$ if
     $\D$ can solve all the time-free problems solvable by~$\D'$.
In this paper we compare the implementability-based hierarchy of failure detectors to the
     hierarchy based on  solvability.
This is done by introducing a new proof technique for establishing the
     solvability relation.
We apply this technique to known failure detectors from the literature
     and demonstrate significant differences between the hierarchies.
 \end{abstract}

\section{Introduction}

Failure detectors \cite{chan:ufdfr} provide an oracular mechanism to
     circumvent  the impossibility of several problems in fault-prone
     asynchronous systems~\cite{FLP,fich:03:hirdc}.
Intuitively, the idea is to enrich asynchronous systems with
     information about failures that may be useful to overcome the
     difficulties posed by process crashes.
Chandra and Toueg~\cite{chan:ufdfr} and Chandra, Hadzilacos, and
     Toueg~\cite{chan:twfdf1} demonstrated landmark results relating to failure detectors: %.
the results in \cite{chan:ufdfr} demonstrated the use of failure
     detectors to solve consensus and other related problems, while
     the results in \cite{chan:twfdf1} showed that any failure
     detector that can be used to solve consensus can also be used to
     implement a failure detector called~$\Omega$.
Since~$\Omega$ is also sufficient to solve consensus, it is the
     weakest failure detector to solve consensus.
To arrive at these important results, \cite{chan:ufdfr} and
     \cite{chan:twfdf1} introduced a relation to compare the ``power''
     of failure detectors: %.
denoted by $\D\redCT \D'$, a failure detector $\D$ to said to be
     stronger than $\D'$, if $\D$ can be used to \emph{implement}
     $\D'$.
     
Since \cite{chan:twfdf1}, the relation $\redCT$ has been used to prove
     similar results for several other problems and has motivated the
     view that failure detectors could be used as ``computability
     benchmark'' \cite{FGK11}; that is, an answer to the question on
     the weakest failure detector  to solve a problem $\problem$ is
     said  to provide the minimal synchrony assumptions necessary to
     solve $\problem$ in fault-prone systems \cite{chan:twfdf1,FGK11}.
This viewpoint is based on several implicit assumptions, one of which
     is that the  hierarchy of failure detectors induced by the
     relation~$\redCT$ is similar to hierarchies induced by other
     natural relations, in other words, that it is robust.

In the work presented here, we focus on this assumption and explore
     the nature of relations that compare failure detectors.
Incidentally, the robustness of the $\redCT$ relation has been
     challenged in recent work
     \cite{jayanti:ephawfd,charron-bost:10:isolt,cornejoetalAFD,cornejoetalAFD-TR}, where it was
     observed that the relation has several drawbacks; for instance,
     $\redCT$ is not reflexive.
To overcome the drawbacks of the~$\redCT$ relation, new relations have
     been proposed in \cite{jayanti:ephawfd} and
     \cite{charron-bost:10:isolt}.

Jayanti and Toueg introduced a new relation in \cite{jayanti:ephawfd},
     which we denote $\redJT$, with a different notion of what it
     means to \emph{implement} a failure detector.
The new relation $\redJT$ extends $\redCT$ and avoids several
     drawbacks of the $\redCT$ relation\footnote{We provide detailed
     descriptions of the $\redCT$ and $\redJT$ relations in Section
     \ref{sec:comparisonRelations}.}.
Based on the $\redJT$ relation, Jayanti and Toueg then demonstrate
     that every problem has a weakest failure detector.
The results in \cite{jayanti:ephawfd} actually holds true for a
     specific class of problems, and in fact, later work by Bhatt and
     Jayanti~\cite{bhatt:oteow} shows that there exist a different
     class of problems that do not have a weakest failure detector.
The apparent contradiction\footnote{There is no real
     contradiction here.
The reconciliation between \cite{jayanti:ephawfd} and
     \cite{bhatt:oteow} is explained in \cite{bhatt:oteow}.} between
     \cite{jayanti:ephawfd} and \cite{bhatt:oteow} regarding the
     existence of weakest failure detectors demonstrates the
     significant dependence of weakest failure detector results on the
     definition of a ``problem''  and choice of the failure detector
     comparison relation.

In \cite{charron-bost:10:isolt}, Charron-Bost et al.\  advocate a
     new comparison relation denoted by~$\redSolv$.
By definition, $\D \redSolv \D'$ if  every (time-free)
     problem solvable by $\D'$ is also solvable by $\D$.
In contrast to the $\redCT$ and $\redJT$ relations, which are based on
     implementing one failure detector using another, the $\redSolv$
     relation depends on the set of problems solvable by each failure
     detector.
If $\D \redCT \D'$, or $\D \redJT \D'$, then any problem
     solvable with~$\D'$ can be solved using $\D$.
Consequently, it is straightforward that~$\redSolv$ extends $\redCT$
     and~$\redJT$.
However, given two failure detectors $\D$ and $\D'$,
     \cite{charron-bost:10:isolt} provides no mechanism for
     demonstrating $\D \redSolv\D'$ without having to establish $\D
     \redCT \D'$ or $\D \redJT \D'$.
In effect, it is not clear how the $\redSolv$ relation differs from
     the $\redCT$ and $\redJT$ relations.

\medskip

\paragraph{Summary of results.} 

In this paper, we address the aforementioned issues by providing a new
     proof technique to establish the $\redSolv$ relation.
Our approach is based on algorithm transformations, and to our
     knowledge, we are first to do so in the context of failure
     detector comparison.
Although, from a technical viewpoint, the proofs are similar to
     existing proofs that establish $\redCT$ and $\redJT$ relations,
     the relationships resulting from our proof technique differ
     significantly from existing relationships among some failure
     detectors.
     
In order to illustrate the difference between $\redCT$ and $\redSolv$,
     we consider three families of failure detectors: the perfect
     failure detector $\PFD$ \cite{chan:ufdfr}, the Marabout failure
     detector $\MFD$ \cite{guer:01:hfap}, and the~$\PFD_{k}$ sequence
     \cite{bhatt:oteow}.\footnote{We describe these failure detectors
     in detail and give their definitions in
     Section~\ref{sec:examples}.} 

The results in \cite{guer:01:hfap} established that $\MFD \not\redCT
     \PFD$ and $\PFD \not\redCT \MFD$.
In contrast, we show that $\MFD$ may be used to solve all the problems
     solvable using $\PFD$, and furthermore, there are problems that
     are solvable using $\MFD$ but not solvable using~$\PFD$.
In other words, we show that $\MFD \redSolv \PFD$ and $\PFD
     \not\redSolv\MFD$.

The results in \cite{bhatt:oteow} show that $\PFD^k \redCT \PFD^{k+1}$
     and $\PFD^{k+1} \not\redCT \PFD^k$.
In contrast, we show that $\PFD^k$ and $\PFD^{k+1}$ can be used to
     solve the same set of problems, that is, for any $k$, the failure
     detectors  $\PFD^k$ and $\PFD^{k+1}$ are equivalent with respect
     to the $\redSolv$ relation.

The results in \cite{guer:01:hfap} and \cite{bhatt:oteow} employ  the
     $\redCT$ relation\footnote{Using arguments similar to the ones
     presented in \cite{guer:01:hfap} and \cite{bhatt:oteow}, one can
     easily show that $\MFD$ and~$\PFD$ are incomparable and $\PFD^k$
     is strictly stronger than $\PFD^{k+1}$  with respect to the
     $\redJT$ relation as well.} to prove that certain failure detectors
     cannot be the weakest ones to solve the given problem.
In contrast, our results show that these conclusions drawn in
     \cite{guer:01:hfap} and \cite{bhatt:oteow} do not hold if the
     failure detectors are compared using the $\redSolv$ relation.
Thus, different natural relations to compare failure detectors lead to
     significantly different results.

\section{The failure detector model}\label{sec:FDmodel}

We recall the basic definitions of the failure detector model
     \cite{chan:ufdfr}.
Informally, it consists of a set of crash-prone processes that are
     connected via reliable asynchronous links and have access to a failure-detector
     oracle that provides information.
In this paper, we only consider failure detectors where this
     information has the form of a subset of the processes in
     the~system.

More formally, the system consists of a finite set of \emph{processes}
     $\Pi$.
We assume that each process $p_i$ in $\Pi$ has a link $l_{(i,j)}$ to
     every process $p_j$ in $\Pi$ over which \emph{messages} can be
     sent.
There is a discrete global time base~$\T$, and for simplicity we
     assume its range of values is the natural numbers $\Naturals$.

\paragraph{Failures and failure patterns.}  

A \emph{failure pattern} is a function $F\colon\T\rightarrow2^\Pi$.
This means that if $p_i\in F(t)$ then $p_i$ has failed by time $t$.
We consider crash faults only, and so $F(t)\subseteq F(t+1)$, for all
     times~$t$.
We say that $p_i$ is live at time $t$ if $p_i\not\in F(t)$, and define
     the set of live processes at time $t$ as
     $\ident{live}(F,t)=\Pi\setminus F(t)$.
A process $p_i$ is correct in $F$ if $p_i$ is always live, that is,
     $p_i \in \ident{correct}(F)=\bigcap_{t\in\T}\ident{live}(F,t)$.
We say processes that are not correct are faulty\dash---or
     crashed\dash---and we abbreviate $\ident{faulty}(F) =
     \Pi\setminus \ident{correct}(F)$.
An \emn{environment} $\E$ is defined as a non-empty set of failure
     patterns.
In this paper, we consider the environment that consists of all
     failure patterns for $\Pi$.

\paragraph{Failure detectors.} 

A failure detector history $H$ is a function $H\colon\Pi \times \T
     \rightarrow 2^\Pi$.\footnote{The failure detectors considered in
     this paper always output a set of processes.
So we do not need the more general original
     definition~\cite{chan:ufdfr} here.} If $\mathcal{H}$ denotes the
     set of all possible histories, then a \emph{failure detector} is
     a function $\D\colon \E \rightarrow 2^\mathcal{H} \setminus
     \emptyset$.

\paragraph{States and configurations.} 

Each process is modeled as a (possibly infinite) state machine $A_i$
     over the set of states $Q_i$ for each process $p_i \in \Pi$.
An algorithm $A$ is a collection of all such state machines
     $(A_i)_{p_i \in \Pi}$.
There exists a non-empty set of states $\hat{Q}_i \subseteq Q_i$
     that are the \emph{initial states} of $p_i$.

Each communication link $l_{(i,j)}$ is also represented by a set of
     states, and the state of each link $l_{(i,j)}$, denoted
     $s_{(i,j)}$ is the set of messages in transit from $p_i$ to
     $p_j$.
The state of a link with no messages in transit is said to be the
     \emph{initial state} of the link.

The \emph{configuration} of a system is a vector $C =
     (s_0,\ldots,s_{n-1},s_{(0,0)},s_{(0,1)}, \ldots, s_{(n-1,n-1)})$
     where $s_i$ is the state of $p_i$ and $s_{(i,j)}$ is the state of
     the link $l_{(i,j)}$.
Then a configuration in which all the processes and links are in
     initial states is called an \emph{initial configuration}.
The set of all configurations of a system is denoted $\mathcal{C}$ and
     the set of all initial configurations is denoted
     $\mathcal{I}$.
The notation $C|_i$ denotes the state of $p_i$ in configuration $C$,
     and $C|_\Pi$ denotes the vector of states of the processes in
     $\Pi$ in configuration $C$.
Similarly, the notation $C|_{(i,j)}$ denotes the state of the link
     $l_{(i,j)}$ in configuration $C$.

\paragraph{Steps.} 

Each transition of the state machine $A_i$\dash---or \emph{step} of
     the process $p_i$\dash---takes as input the current state $s$ of
     the process, zero or one message $m_r$ (the ``received''
     message), and an output~$d$ from the failure detector; it
     produces as output a new state $s'$ for the process and may send
     a message $m_s$ to another process $p_k$ via the corresponding
     communication link (the ``sent'' message).
Incidentally, the receipt of a message by a process $p_i$ from $p_j$
     removes the message from the link $l_{(j,i)}$ and the sending of
     a message by $p_i$ to $p_k$ adds the message to the link
     $l_{(i,k)}$; %.
this step can then be identified by the tuple $(p_i, s, m, d,
     s', m')$, where $m$ is~$\bot$ if no message is received
     and  $(p_j, m_r)$ otherwise, and similarly, $m'$ is  $\bot$ if no
     message is sent and  $(p_k, m_s)$ otherwise.

\paragraph{Schedules.} 

A schedule $\Phi$ of an algorithm $A$ is a sequence of steps taken by
     processes executing $A$; the $\ell$th step of $\Phi$ is denoted
     $\Phi[\ell]$.
A \emph{projection} of a schedule $\Phi$ over a process~$p_i$ is the
     subsequence of $\Phi$ consisting of only the steps executed by
     $p_i$ and is denoted~$\Phi|_{i}$.

\paragraph{Time-Sequences.} 

A time-sequence $T$ is a sequence of increasing values in $\T$; the
     $\ell$th element in $T$ is denoted $T[\ell]$ (which represents
     the time at which the step $\Phi[\ell]$ occurs).
Again, we define a \emph{projection} of a time-sequence $T$ over a
     process $p_i$ as a subsequence of $T$ consisting of only the
     times at which $p_i$ executes steps and is denoted~$T|_i$.

\paragraph{Runs.}

A run $R$ of an algorithm $A$ using a failure detector $\D$ is a tuple
     $\langle F, H, I, \Phi, T \rangle$, where $F$ is a failure
     pattern, $H\in\D(F)$ is a failure detector history, $I \in
     \mathcal{I}$ is an initial configuration of~$A$, $\Phi$ is a
     schedule of $A$, and $T$ is a time-sequence.
Run $R$ is \emph{valid for $A$}\dash---or just \emph{valid} for
     short\dash---if correct processes take an infinite number of
     steps and if
 for each $\ell \ge 1$, the step $\Phi[\ell]
     \equiv (p_i,s,m,d,s',m')$ satisfies the following properties.      
\begin{itemize}

\item The process $p_i$ is live at time $T[\ell]$; that is, $p_i
     \notin F(T[\ell])$.

\item $d$ is an output of the failure detector $\D$ at time $T[\ell]$;
     formally, $d = H(p_i,T[\ell])$.

\item There are no spurious messages, that is, if $m$ is of the form
     $(p_j,m_r)$, then there exists some $k<\ell$ such that $m_r$ is a
     message that was sent by $p_j$ to $p_i$ in step $\Phi[k]$ identified by
     $(p_j,*,*,*,*,(p_i,m_r))$.

\item Message transmission is reliable, that is, if $m'$ is of the
     form $(p_j,m_s)$, then there is at most one $k>\ell$ such that
     step $\Phi[k]$ is of the form $(p_j,*,(p_i,m_s),*,*,*)$.
Furthermore, if $p_j$ is correct, then there is exactly one such step.

\item If $\Phi[\ell]$ is the first step of process in $p_i$ in run $R$, then
     $s=I|_i$.

\item The state of a process does not change between consecutive steps by that process; that
     is, if $p_i$ takes another step, then the first step of $p_i$
     after $\Phi[\ell]$  is of the form $(p_i,s',*,*,*,*)$.
\end{itemize}

\paragraph{Configuration sequences induced by runs.}
     Given a run $R = \langle F, H, I, \Phi, T \rangle$, the
     configuration of the system after $k$ steps are taken is given by
     $\gamma(I,\Phi,k)$.
The sequence $\gamma(I,\Phi,0), \gamma(I,\Phi,1), \dots$ is the
     \emph{configuration sequence of run}~$R$.
The state of process~$p_i$ after~$p_i$ takes $k$ steps in the run is
     given by $\gamma_i(I,\Phi,k)$; if process $p_i$ crashes and takes
     only~$k$ steps, then we use the convention that
     $\gamma_i(I,\Phi,\ell) = \gamma_i(I,\Phi,k)$ for $\ell\ge k$.

Note that if two runs share the same $I$ and $\Phi$ (but differ, for
     instance, at the times steps are taken), then they induce the
     same configuration sequence.

\section{Solving problems}\label{sec:problemDefinition}

We now define the notion of a \emph{problem} and what it means to
     \emph{solve} a problem.
Problems traditionally depend on initial values (as in consensus
     \cite{FLP}) and transitions to certain states depending on the
     initial values.
So we have to define a problem by referring to problem states.
Problems also depend on the correctness of processes.
For instance, faulty processes are not required to make progress.
In the failure-detector model, faults are modeled by failure patterns,
     which define after what time faulty processes must not take
     steps.
However, before that, processes need not take steps.
As we want to get rid of all time dependencies in the problem
     definition, it is hence natural to restrict problems by the set
     of processes that appear in the failure pattern rather than
     restricting the problems by the times at which  processes appear
     in the failure pattern.
This is done in the \emph{crash time independence} property described
     later.

Moreover, as we define problems to be solvable in asynchronous
     systems, we have to consider the nature of runs in such systems.
Since message delays and process speeds are unconstrained in
     asynchronous systems, processes may take finitely many idempotent
     or no-op steps while waiting for a message, or while waiting on
     some local predicate to become true.
To reflect this, we require that problems are tolerant to \emph{finite
     stuttering} which is described after the following preliminary
     definitions.

We start by defining $\sigma$ as a set of \emph{problem states}.
By $\hat{\sigma}$ we denote the set of \emph{initial problem states},
     with $\hat{\sigma}\subseteq\sigma$.
A \emph{problem configuration} $\conf$ for a system of size $n$ is an
     $n$-dimensional vector of problem states.
We denote by $\conf|_i$, the problem state associated with process
     $p_i$ in the problem configuration $\conf$.
A problem configuration consisting only of initial problem states is
     called an \emph{initial problem configuration} $\initConf$.
We denote $\setConf$ to be the set of all possible problem
     configurations, and we denote $\initSetConf$ to be the set of all
     possible initial problem configurations; note that $\initSetConf
     \subseteq \setConf$.
We denote $\seqprobconf$ to be the set of all sequences of problem
     configurations that start with an initial problem configuration.

Further, let $\seq_{\ident{pre}}$ be a finite problem configuration
     sequence starting with an initial problem configuration,  let
     $\seq_{\ident{suff}}$ be a problem configuration sequence, and
     let $\conf$ and $\conf'$ be two problem configurations.
Let $\conf_{mid}$ be any problem configuration such that for each
     process $p_i$, either $\conf_{mid}|_i = \conf|_i$ or
     $\conf_{mid}|_i = \conf'|_i$.
Then, for any problem configuration sequence $\seq =
     \seq_{\ident{pre}} \cdot \conf \cdot \conf' \cdot
     \seq_{\ident{suff}}$, the sequence $\seq' = \seq_{\ident{pre}}
     \cdot \conf \cdot \conf_{mid} \cdot \conf' \cdot
     \seq_{\ident{suff}}$ is a \emph{$1$-stutter} of $\seq$ denoted by
     $\seq \sqsubset_1 \seq'$.
Inductively for each $n>1$, we define $\seq'$ to be an
     \emph{$n$-stutter} of~$\seq$, denoted by $\seq \sqsubset_{n}
     \seq'$, if there is a sequence $v$ such that $\seq
     \sqsubset_{n-1} v \; \wedge \; v \sqsubset_1 \seq'$.
Further, we define~$\seq'$ to be a \emph{stutter} of $\seq$, denoted by
     $\seq \sqsubseteq \seq'$, if either $\seq=\seq'$ or there is an
     $n$, $0<n<\infty$, such that $\seq \sqsubset_{n} \seq'$.

\paragraph{Problems.}
Briefly, a problem is a predicate over a problem configuration
     sequence that starts with an initial problem configuration, and a
     fault pattern.
More precisely, a \emph{time-free problem} $\problem$ over
     $\seqprobconf$  
%for a system consisting of the processes in~$\Pi$ 
in fault
     environment~$\E$\dash---or just \emph{problem} for short\dash---
     is a predicate $\problem$ on $\seqprobconf \times \E$ with
     the following properties: 
\begin{itemize}

\item \emph{Crash time independence.}   For all failure patterns $F$ and $F'$ in
     $\E$ and for all $\seq$ in~$\seqprobconf$,  $correct(F) =
     correct(F')$ implies $\problem(\seq,F) = \problem(\seq,F')$.

\item \emph{Finite stuttering.} For any failure pattern $F$, and any
     two problem configuration sequences $\seq$ and $\seq'$ in
     $\seqprobconf$, $\seq \sqsubseteq \seq'$ implies
     $\problem(\seq,F)= \problem(\seq',F)$.
\end{itemize}

\paragraph{Solving a problem}

Let $\Alg$ be an algorithm, and let a problem $\problem$ be defined
     for $\seqprobconf$ and $\E$.
Let an \emph{interpretation} $V_i$ be a function that maps the states
     $Q_i$ of $\Alg$ to $\sigma$ (the problem states that constitute
     $\seqprobconf$), such that the initial states of the
     algorithm~$\hat{Q}_i$ are mapped onto $\hat{\sigma}$
     (surjective).
This naturally extends to a function $V_\Pi$ that maps configurations
     $C|_\Pi$ to problem configurations.
An \emph{interpreted run} is a sequence of problem configurations
     obtained by applying  $V_\Pi$ to the configuration sequence of a
     valid run $R=\langle F, H, I, \Phi, T \rangle$ of~$\Alg$; it is
     denoted by $\ident{ir}(R,V_\Pi)$.
Further, the set of all interpreted runs of algorithm $\Alg$
     using~$\D$ with failure pattern $F$ interpreted by $V_\Pi$ is
     denoted by $\ident{IR}(\Alg,F,\D,V_\Pi)$.

Algorithm $\Alg$ solves a problem $\problem$ using failure detector
     $\D$ in environment $\E$, if there is a function $V_\Pi$  such
     that for all $F$ in $\E$ and any $w \in
     \ident{IR}(\Alg,F,\D,V_\Pi)$,  the predicate $P(w,F)$ holds.
If there is an algorithm that solves problem $\problem$ using failure
     detector $\D$ we say that failure detector~$\D$ \emph{can be used
     to solve} $\problem$, or in other words $\problem$ \emph{is
     solvable using} $\D$.

\medskip

The definition of a problem encompasses many common problems in distributed computing, including classic agreement problems.
The set of problem states of consensus, for instance, can be defined
     as $\sigma = \left\{(p,d)\colon p\in\{0,1\} \wedge
     d\in\{\bot,0,1\}\right\}$. A problem state $(p,d)$ at process $p_i$ signifies a state where a process $p_i$ has
     $p$ as its proposed initial value, and $d$ is its decision; if $p_i$ has not yet decided, then 
     $d=\bot$, and otherwise $d$ is $p_i$'s final decision.
The set of initial problem configurations $\initSetConf$ is the set of
     all $n$-element vectors where each $i$-th element is a problem state of $p_i$ and is of the form $(p,\bot) \in \sigma$.
One can then naturally define the consensus properties agreement,
     termination, and validity as predicates on problem configuration
     sequences, and consensus as the conjunction of these predicates.

\section{Comparison relations}\label{sec:comparisonRelations}

\paragraph{Chandra-Toueg relation.}  

We recall from~\cite{chan:ufdfr,chan:twfdf1} that $\D\redCT\D'$ is
     defined via failure detector transformation as follows.
An  algorithm $\Trans{\D\rightarrow\D'}$ uses $\D$ to maintain a
     variable $\ident{out}_i$ at every process~$p_i$.
This variable emulates the output of $\D'$ at $p_i$.
Let $O_R$ be the history of all the $\ident{out}_i$ variables in run
     $R$, that is, $O_R(p_i,t)$ is the value of $\ident{out}_i$ at
     time~$t$ in run~$R$.
Algorithm $\Trans{\D\rightarrow\D'}$ \emph{transforms}~$\D$
     \emph{into} $\D'$ if for every valid run $R=\langle F, H , I,
     \Phi, T \rangle$ of $\Trans{\D\rightarrow\D'}$ using $\D$,
     $O_R\in \D'(F)$.
If such an algorithm $A$ exists, then $\D\redCT\D'$.

\paragraph{Jayanti-Toueg relation.}  

The relation $\redJT$, introduced in \cite{jayanti:ephawfd}, differs
     from $\redCT$ in that the notion of what it means to transform a
     failure detector is different from the one used
     in~\cite{chan:ufdfr}; partly by changing the computational model.
Instead of using the failure detector value at the time the step
     occurs, the ``query mechanism'' is modeled via a query to the
     failure detector at time~$t$ and a response from the failure
     detector at some time $t'>t$.
Specifically, an algorithm $\Trans{\D\rightarrow\D'}$ uses $\D$ and
     transforms $\D$ to $\D'$ if and only if, for every valid run of
     $\Trans{\D\rightarrow\D'}$, there exists a history $H$ of $\D'$
     under the failure pattern of the run such that the following is
     true.
For each process~$p_i$, and for each query by $p_i$ to
     $\Trans{\D\rightarrow\D'}$ which happens at some time $t$,
     $\Trans{\D\rightarrow\D'}$ responds with an output $out$ at some
     time $t'\geq t$, and $out \in \{H(p_i,s)\colon s \in [t,t']\}$.
Hence, the definition of transformation does not require  maintaining
     a variable $\ident{out}_i$ but rather requires ensuring
     consistency of the query and response events.

\paragraph{Solvability relation.} 

The relation $\redSolv$, introduced in \cite{charron-bost:10:isolt},
     states that a failure detector $\D$ is stronger than $\D^\prime$
     with respect to the solvability relation, denoted
     $\D\redSolv\D^\prime$, if~$\D$ {can be used to solve} any problem
     solvable using $\D^\prime$.

\bigskip

%% \paragraph{Proving relationships.}

The definitions of $\redCT$ and $\redJT$ provide a straightforward
     proof technique to demonstrate the claims $\D \redCT \D'$ and $\D
     \redJT \D'$.
In order to prove $\D \redCT \D'$ or  $\D \redJT \D'$ one has to
     provide an algorithm $\Trans{\D\rightarrow\D'}$ that has the
     properties described above.

If $\D\redCT\D'$ then every problem solvable with $\D'$ is solvable
     with $\D$ \cite{chan:ufdfr,chan:twfdf1} and thus  $\redSolv$
     extends $\redCT$.
Similarly, one sees that $\redSolv$ extends $\redJT$ as well.
However, if  $\D\not\redCT\D'$, no proof technique has been given so
     far to establish $\D\redSolv\D'$.

\section{New technique for proving the solvability relation} 

Our approach is based on the following idea.
If a problem $\problem$ is solvable using $\D'$, then there exists an
     algorithm $\Alg$ that uses $\D'$ and solves $\problem$.
If we can transform $\Alg$ to another algorithm $\TrAlg$ such that
     $\TrAlg$ uses $\D$ and solves $\problem$, then we have shown that
     problem $\problem$ is also solvable using $\D$.
Furthermore, if we demonstrate the aforementioned result for every
     problem solvable using $\D'$, then we have shown that
     $\D\redSolv\D'$.

More generally, the proof technique focuses on defining a
     transformation function~$\mathfrak{F}$ whose domain is the set of
     all algorithms that use $\D'$ and whose range is the set of
     algorithms that use $\D$ such that if algorithm $\Alg$  uses
     $\D'$ to solve $\problem$, then $\mathfrak{F}(\Alg)$ uses $\D$
     and solves $\problem$.

In order to prove that the function~$\mathfrak{F}$ actually has this
     desired property, we consider an arbitrary problem $\problem$
     solvable using $\D'$.
We do so by considering an algorithm $\Alg$ that solves $\problem$
     using $\D'$.
By definition, such an algorithm must exist.
Moreover, there is a function $V_\Pi$ which maps configurations of
     each valid run $R$ of $\Alg$ using $\D'$ to a sequence of problem
     configurations that satisfy $\problem$.
Using $V_\Pi$, we define a new function $\widetilde{V}_\Pi$ that maps
     the configurations of $\mathfrak{F}(\Alg)$ to problem
     configurations.
We then have to show that for any interpreted run $w \in
     \ident{IR}(\mathfrak{F}(\Alg),F,\D,\tilde{V}_\Pi)$,  the
     predicate $P(w,F)$ holds.

\section{Failure detectors under consideration}\label{sec:examples}

\subsection{Definitions}

In this section we define the three kinds of failure detectors that we
     are going to use in this paper.
The \emph{perfect failure detector} $\PFD$ was originally proposed in
     \cite{chan:ufdfr}.
Informally, $\PFD$ eventually and permanently suspects crashed
     processes and never suspects live processes.
More precisely, $\PFD$ is defined to ensure \emph{strong completeness}:   
$$
\forall F\in\E,\;\forall H\in\PFD(F),\; 
  \forall p_j\in\ident{faulty}(F),\;
\forall p_i \in \ident{correct}(F), \;
\exists t' \in \T, \;
\forall t> t'\colon \;
p_j\in H(p_i,t),
$$
and \emph{strong accuracy}:
$$
\forall F\in\E,\;\forall H\in\PFD(F),\;\forall t \in \T,\; 
\forall p_i,p_j \in \ident{live}(F,t)
\colon\;p_j\not\in H(p_i,t).
$$

The \emph{Marabout} failure detector $\MFD$ was introduced in
     \cite{guer:01:hfap}\footnote{Although the definition printed in
     \cite{guer:01:hfap} is slightly different (only failure detector
     outputs of correct processes instead of live processes are
     restricted), we claim that actually the definition given here is
     used in the proof sketches in \cite{guer:01:hfap}.
Otherwise, for instance, the proof sketch of
     \cite[Proposition~3.3]{guer:01:hfap} would fail; one could easily
     construct a case where a process that is going to crash in the
     future decides differently from a correct process.},
and it always outputs the set of faulty processes.
It is defined as:  
$$
 \forall F\in\E,\;
\forall H\in\MFD(F),\;
\forall t \in \T,\; 
\forall p_i \in \ident{live}(F,t): \;
 H(p_i,t) = \ident{faulty}(F).
$$

The $\PFD_k$ failure detector was introduced in \cite{bhatt:oteow}
     (using the notation ``$\D_k$'' which we find somewhat
     inconsistent with the rest of our notations).
Informally, $\PFD_k$ can provide arbitrary information about processes
     that crash before or at time $k$.
For correct processes and processes that crash after time $k$,
     $\PFD_k$ never suspects these processes before they crash, and
     $\PFD_k$ eventually and permanently suspects these processes
     after they crash.
Formally, $\PFD_k$ satisfies the properties $k$-Completeness: 
\begin{multline}
\forall F \in \E,\; 
\forall H \in \PFD_k(F),\;
\forall p_i,p_j \in \Pi,\;
\exists t' \in \T, \;
\forall t>t' \; 
 \colon \\
(p_j \in \ident{live}(F,k) \wedge p_j \in 
\ident{faulty}(F ) \wedge p_i \in correct
 (F) ) \Rightarrow p_j \in H(p_i, t), \nonumber
\end{multline}
and $k$-Accuracy:
$$\forall F \in \E,\; 
\forall H \in \PFD_k(F),\;
\forall p_i,p_j \in \Pi,\;
\forall t \in \T\ \colon
 (p_j\in \ident{live}(F,k)
\wedge p_j \notin F (t)) \Rightarrow p_j \notin H(p_i, t).$$

\subsection{Comparing $\MFD$ and $\PFD$.}

In \cite{guer:01:hfap} it was shown that $\PFD$ and $\MFD$ are not
     comparable with respect to $\redCT$.
Informally, the arguments for the result are as follows.
No algorithm can tell by message exchange or from looking at the
     output of~$\PFD$ at a certain time which processes will
     eventually crash (in the future), therefore $\PFD\not\redCT\MFD$.
For showing $\MFD\not\redCT\PFD$, note that faulty processes should
     not be put into the set of suspected processes too early by
     $\PFD$, as this would violate strong accuracy.
However, by strong completeness of $\PFD$, crashed processes have to
     be added to the set eventually.
The outputs of $\MFD$ do not allow us to reconcile these two
     requirements.
Hence, no algorithm that queries $\MFD$ can implement $\PFD$; in other
     words, $\MFD\not\redCT\PFD$.
Similar arguments also apply to the $\redJT$ relation, and it can be
     shown that $\MFD$ and $\PFD$ are incomparable with respect to the
     $\redJT$ relation as well.

In this paper, we show for the solvability relation, that
     $\PFD\not\redSolv\MFD$ and $\MFD\redSolv\PFD$.
Demonstrating $\PFD\not\redSolv\MFD$ is straightforward.
It is sufficient to give a problem solvable using $\MFD$ and not
     solvable using $\PFD$.
Consider the following variant of consensus, called \emn{strong
     consensus}, which requires that all the correct processes have to
     output the input value of some unique \emph{correct} process in
     the system, if there is a correct process, and otherwise output
     anything.

Solving this problem using $\MFD$ is straightforward.
Each process sends its input to all the processes and waits for inputs
     from  the set of processes not suspected by $\MFD$.
Since the processes not suspected by $\MFD$ are the correct processes,
     if each process decides on the input of the correct process with
     the smallest ID, the problem is solved.
However, as $\PFD$ does not provide information on process crashes in
     the future, we can show that there is no algorithm that solves
     strong consensus using $\PFD$.
So we conclude that $\PFD\not\redSolv\MFD$.

In order to establish that $\MFD$ is strictly stronger than $\PFD$, it
     remains to show that $\MFD\redSolv\PFD$.
We shall do so in Section~\ref{sec:MP} in which we introduce a general
     transformation \emph{Stall-on-Suspect} that transforms any
     algorithm $A$ using $\PFD$ into an algorithm $\tilde{A}$ using
     $\MFD$.
Intuitively, Stall-on-Suspect ensures that faulty processes do not
     participate in the algorithm.
Given an algorithm $\Alg$, each process first queries $\MFD$ to
     determine whether it is correct or faulty.
If a process $p_i$ queries $\MFD$ and discovers that it is faulty,
     then $p_i$ stops participating in the algorithm by performing
     only no-op steps and sends no messages until it crashes.
Otherwise, process~$p_i$ follows the original algorithm $\Alg$
     faithfully.
We show in Section~\ref{sec:MP} that each valid run of the modified
     algorithm using $\MFD$ is indistinguishable from some valid run
     of the original algorithm using $\PFD$ where faulty processes
     crash initially, at time $0$.
Since, by assumption, the original algorithm solves the  problem using
     $\PFD$, the same problem is solvable by $\MFD$ as well.
Thus, we show that every problem solvable by $\PFD$ is also solvable
     by $\MFD$.

\subsection{Comparing $\PFD_k$ failure detectors}

In \cite{bhatt:oteow}, the series of $\PFD_k$ failure detectors were
     proposed to solve FCFS mutual exclusion.
Note that various values of $k$ instantiate different failure
     detectors, and it was shown in \cite{bhatt:oteow} for all $k\ge
     0$ that $\PFD_{k} \redCT \PFD_{k+1}$ and $\PFD_{k+1}
     \not\redCT \PFD_{k}$.
The proof of the former is based on the observation that the trivial
     transformation (namely, at each step, write the current failure detector
     output into $\ident{out}_i$) is sufficient to implement
     $\PFD_{k+1}$ using $\PFD_k$; intuitively, correctness follows because the histories of
     $\PFD_k$ are a strict subset of the histories of
     $\PFD_{k+1}$.\footnote{This argument is in general not sufficient
     to prove $\redCT$ as shown in~\cite{charron-bost:10:isolt}.
It works in this case, as $\PFD_k$ belongs to the class of failure
     detectors called ``time-free'' in~\cite{charron-bost:10:isolt};
     they allow finite stuttering.}

The latter ($\PFD_{k+1} \not\redCT \PFD_{k}$) is established by
     showing that no algorithm that queries $\PFD_{k+1}$ can reliably
     detect if some process has crashed at time $k+1$, which is a
     necessary requirement to implement $\PFD_k$.
Similar arguments show  for all $k\ge 0$ that $(\PFD_{k} \redJT
     \PFD_{k+1})$ and $(\PFD_{k+1} \not\redJT \PFD_{k})$

In this paper, we show for all $k\ge 0$ that $(\PFD_{k} \redSolv
     \PFD_{k+1})\, \wedge\, (\PFD_{k+1} \redSolv \PFD_{k})$.
Demonstrating $\PFD_{k} \redSolv \PFD_{k+1}$ is straightforward and it
     follows from the result $\PFD_{k} \redCT \PFD_{k+1}$ from
     \cite{bhatt:oteow} and the observation that $\redSolv$ extends
     $\redCT$ \cite{chan:ufdfr}.

Therefore, it remains to be shown that $\PFD_{k+1} \redSolv \PFD_{k}$.
We do so in Section~\ref{sec:DK} using a general transformation
     \emph{Delay-a-Step} which just adds a no-op step at the beginning
     of each execution for each algorithm.
Given an algorithm $\Alg$ that solves some problem $\problem$ using
     failure detector $\PFD_{k}$, in the delay-a-step transformation,
     each process $p_i$ first executes a no-op step in which $p_i$
     neither receives nor sends any message; thereafter, $p_i$
     executes the algorithm $\Alg$ but queries $\PFD_{k+1}$ instead of
     $\PFD_k$.
We show in Section \ref{sec:DK} that each valid run of the modified
     algorithm using $\PFD_{k+1}$ induces an interpreted run that is
     also an interpreted run (with ``shifted'' failure pattern) of the
     original algorithm using $\PFD_k$.
Since, by assumption, the original algorithm solves $\problem$ using
     $\PFD_k$, problem $\problem$ is solvable by $\PFD_{k+1}$ as well.
Thus, we show that every problem solvable by $\PFD_k$ is also solvable
     by $\PFD_{k+1}$.

\section{Every problem solvable using $\PFD$ is solvable using $\MFD$}
\label{sec:MP}

\subsection{Algorithmic transformation: Stall-on-Suspect}
\label{subsec:algTrans}

Informally, the \emph{Stall-on-Suspect} transformation (SoS) converts
     an algorithm $\Alg$ to an algorithm $\TrAlg$ such that $\TrAlg$
     at a process $p_i$ behaves exactly like $\Alg$ if the failure
     detector at $p_i$ does not suspect itself initially.
Otherwise, $\TrAlg$ goes into a special stall state in which it
     remains for the remainder of the execution.
 
More precisely, the SoS transformation is defined by a function
     $\mathfrak{F}_{SoS}(\Alg)$ that maps an algorithm $\Alg =
     (\Alg_i)_{\forall p_i \in \Pi}$ that uses a failure detector that
     outputs a list of suspected processes to a new algorithm $\TrAlg
     = (\TrAlg_i)_{\forall p_i \in \Pi}$.
The new algorithm $\TrAlg$ is constructed as follows.
First, for each process $p_i$, we add a new set of states
     $S^\dagger_i$ to the states of $\Alg_i$, such that $|S^\dagger_i|
     = |\hat{Q}_i|$.
The states in $S^\dagger_i$ are not initial states in $\TrAlg_i$.
We define a bijective function $\ident{stall}_i\colon \hat{Q}_i
     \rightarrow S^\dagger_i$ that maps the initial states of process
     $p_i$ to states in $S^\dagger_i$.

The state transitions in $\TrAlg_i$ differ only in the transitions
from initial states:
If a process~$p_i$ of~$\TrAlg_i$ is in state $q\in\hat{Q}_i$, and if
     the failure detector output of a step of $p_i$ contains $p_i$,
     then~$p_i$ sends no message and goes into state
     $\ident{stall}_i(q)$.
Otherwise, $p_i$'s step is the one specified by~$\Alg_i$.
If a process $p_i$ of $\TrAlg_i$ is in $s\in S^\dagger_i$, then $p_i$
     sends no message and remains in state $s$ in each step.

\subsection{Solving $\problem$ using $\mathfrak{F}_{SoS}(\Alg)$}

Consider the algorithm $\TrAlg = \mathfrak{F}_{SoS}(\Alg)$.
Let $\widetilde{R}=\langle F, H, I, \widetilde{\Phi}, \widetilde{T}
     \rangle$ be an arbitrary valid run of~$\TrAlg$ using $\MFD$.
Let $\Phi$ and $T$ be the schedule and time sequence obtained by
     removing the entries corresponding to steps of processes in
     $\ident{faulty}(F)$ from $\widetilde{\Phi}$ and $\widetilde{T}$,
     respectively.

\begin{proposition}\label{prop:RvalidMFDrun}
If $\widetilde{R}=\langle F, H, I, \widetilde{\Phi}, \widetilde{T}
     \rangle$ is a valid run of $\TrAlg$ using failure detector
     $\MFD$, then $R=\langle F, H, I, \Phi, T \rangle$ is a valid run
     of $\Alg$ using $\MFD$ where no faulty process takes a step.
\end{proposition}
\begin{proof}
To show this proposition, one has to check that the consistency
     requirements of a valid run from Section \ref{sec:FDmodel} are
     met in~$R$.
Since the output of $\MFD$ at a faulty process always suspects itself,
     in the first step of a faulty process in $\TrAlg$, the process
     transitions to a state in~$S^\dagger$ and never sends a message.
Therefore, faulty processes do not send messages  in
     run~$\widetilde{R}$ of~$\TrAlg$.
Since correct processes never suspect themselves, they take the same
     steps in $R$ and~$\widetilde{R}$ by construction.
Consequently, $R$ does not contain any steps in which a message from a
     faulty process is received.
Apart from this, the consistency of $R$ follows from the consistency
     of~$\widetilde{R}$.
\end{proof}

Given a failure pattern $F$, let $F^0$ be the \emn{initial crash
     scenario}, that is, the failure pattern where  $F^0(0) =
     \ident{faulty}(F)$ and for any $t>0$, $F^0(t)= F^0(0)$.

\begin{proposition}\label{prop:R0validMFDrun}
If $R=\langle F, H, I, \Phi, T \rangle$ is a valid run of $\Alg$ using
     $\MFD$ where no faulty process takes a step, then $R^0=\langle
     F^0, H, I, \Phi, T \rangle$ is a valid run of $\Alg$ using
     $\MFD$.
\end{proposition}
\begin{proof}
We prove this proposition by showing that $R^0$ satisfies the
     consistency conditions of a valid run as specified in Section
     \ref{sec:FDmodel}.
Note that in $R^0$ all faulty processes crash at time~$0$; therefore,
     no faulty process takes a steps in $R^0$.
Since $correct(F)=correct(F^0)$, the history~$H$ is a valid history of
     $\MFD$ for fault pattern $F^0$.
Since $R$ and $R^0$ share the same schedule $\Phi$ and $R$ is a valid
     run of $\Alg$ using $\MFD$, remaining consistency conditions for
     $R^0$ follows from the consistency of $R$.
\end{proof}

From the definition of $\MFD$ and $\PFD$ one observes that in initial
     crash scenarios, the history of $\MFD$ is in the set of allowed
     histories of $\PFD$, and therefore we find: 

\begin{proposition}\label{prop:R0validPFDrun}
If $R^0=\langle F^0, H, I, \Phi, T \rangle$ is a valid run of $\Alg$
     using $\MFD$, then $R^0$ is a valid run of $\Alg$ using $\PFD$.
\end{proposition}

From the three propositions above we infer

\begin{theorem}\label{thm:MPcorr}
For any valid run $\widetilde{R}=\langle F, H, I, \widetilde{\Phi},
     \widetilde{T} \rangle$ of $\TrAlg$ using $\MFD$ there is a valid
     run $R^0=\langle F^0, H, I, \Phi, T \rangle$ of $\Alg$ using
     $\PFD$.
\end{theorem}

Next, we argue that if algorithm $\Alg$ solves problem~$\problem$
     using $\PFD$, then $\TrAlg$ solves $\problem$ using~$\MFD$.
Assuming that $\Alg$ solves $\problem$,  there is an interpretation
     $V_\Pi$ such that for all $F$ in $\E$ and any $w \in
     \ident{IR}(\Alg,F,\D,V_\Pi)$,  the predicate $P(w,F)$ holds.
As any interpreted run of $\Alg$ using $\PFD$ satisfies the problem,
     and since by Theorem~\ref{thm:MPcorr} every valid run of $\TrAlg$
     using $\MFD$ can be mapped to a valid run of  $\Alg$ using
     $\PFD$, we have to show that the mapping from $\widetilde{\Phi}$
     to $\Phi$ ensures that $\TrAlg$ also solves the problem using
     $\MFD$.

To this end, we obtain $\widetilde{V}_\Pi$ by defining for each
     process $p_i$ a new function $\widetilde{V}_i$ as a mapping of
     each state of $p_i$ in $\TrAlg$ to a problem state: for states $s
     \in S^\dagger_i$ we define $\widetilde{V}_i(s) =
     V_i(\ident{stall}_{i}^{-1}(s))$, and for all other states $s$ of
     $p_i$ we define $\widetilde{V}_i(s) = V_i(s)$.

As $\ident{faulty}(F) = \ident{faulty}(F^0)$, we just speak of faulty
     (or correct) processes in the following, as no confusion may
     occur.

\begin{proposition}\label{prop:correctProcessSameVi}
If $\widetilde{R}=\langle F, H, I, \widetilde{\Phi}, \widetilde{T}
     \rangle$  is valid run of $\TrAlg$ using failure detector $\MFD$
     and if~$R^0=\langle F^0, H, I, \Phi, T \rangle$ is a valid run of
     $A$ using $\PFD$, then for any correct process $p_i$ and for any
     index $\ell \ge 0$: 
$$
\widetilde{V}_i(\gamma_i(I,\widetilde{\Phi},\ell)) =
V_i(\gamma_i(I,\Phi,\ell)).
$$
\end{proposition}

\begin{proof}
Since, $\Phi$ is constructed from $\widetilde{\Phi}$ by deleting the
     no-op steps taken by faulty processes, we know that each correct
     process $p_i$ follows the same sequence of states in
     $\widetilde{\Phi}$ and $\Phi$.
That is, $\gamma_i(I,\widetilde{\Phi},\ell) = \gamma_i(I,\Phi,\ell)$.
Since $p_i$ is correct, $p_i$ is never suspected by both $\MFD$ and
     $\PFD$.
Therefore, in $\widetilde{R}$, $p_i$ is never in any state in
     $S^\dagger$.
Hence, for each state $s$ that $p_i$ is in $\widetilde{R}$,
     $\widetilde{V}_i(s) = V_i(s)$.
In other words, $\widetilde{V}_i(\gamma_i(I,\widetilde{\Phi},\ell)) =
     V_i(\gamma_i(I,\Phi,\ell))$.
\end{proof}

\begin{proposition}\label{prop:faultyProcessSameVi}
If $\widetilde{R}=\langle F, H, I, \widetilde{\Phi}, \widetilde{T}
     \rangle$  is valid run of $\TrAlg$ using failure detector $\MFD$
     and if~$R^0=\langle
     F^0, H, I, \Phi, T \rangle$ is a valid run of $A$ using $\PFD$,
     then for any faulty process $p_i$ and for any index $\ell\ge 0$:
$$\widetilde{V}_i(\gamma_i(I,\widetilde{\Phi},\ell)) =
     V_i(\gamma_i(I,\Phi,\ell)).$$ 

\end{proposition} 

\begin{proof}
Since faulty processes do not take any steps in $R^0$, we know that
     for each faulty process $p_i$, and each index $\ell \geq 0$ in
     run $R^0$, $\gamma_i(I,\Phi,\ell) = I|_i$.

In run $\widetilde{R}$, we know from the construction of algorithm
     $\TrAlg$ that each faulty process $p_i$, initially, in state
     $\hat{q}_i \in \hat{Q}_i$, enters a state $s^\dagger_i \in
     S^\dagger_i$ in its first step where $s^\dagger_i =
     \ident{stall}(\hat{q}_i)$, and remains there until it crashes.
Therefore, for each faulty process $p_i$, and each index $\ell \geq 0$
     in run $\widetilde{R}$, $\gamma_i(I,\Phi,\ell) \in \{\hat{q}_i,
     s^\dagger_i\}$.

From the definition of $\widetilde{V}_i$, we know that
     $\widetilde{V}_i(\hat{q}_i) = V_i(\hat{q}_i)$, and
     $\widetilde{V}_i(s^{\dagger}_i) = V_i(stall^{-1}_i(s))$.
As $s^\dagger_i = \ident{stall}(\hat{q}_i)$, we obtain
     $\widetilde{V}_i(s^\dagger_i) = V_i(\hat{q}_i)$.
Therefore, for each faulty process $p_i$, and each index $\ell \geq 0$
     in run $\widetilde{R}$, $\widetilde{V}_i(\gamma_i(I,\Phi,\ell)) =
     V_i(\hat{q}_i)$.

Since each process $p_i$ is in the same initial state in
     $\widetilde{R}$ and $R^0$, we have $\hat{q}_i = I|_i$.
Therefore, $\widetilde{V}_i(\gamma_i(I,\Phi,\ell)) = V_i(\hat{q}_i) =
     V_i(\gamma_i(I,\Phi,\ell))$.
\end{proof}

\begin{theorem} 
If $\Alg$ solves $\problem$ using~$\PFD$ then
     $ \TrAlg=\mathfrak{F}_{SoS}(A)$ solves $\problem$ using $\MFD$.
\end{theorem}
\begin{proof}
 Since  $\Alg$ solves $\problem$ using $\PFD$, we know that there exists a function $V_\Pi$ such that for any $w \in \ident{IR}(\Alg,F,\PFD,V_\Pi)$,  the predicate $P(w,F)$ is true.

Let $\TrAlg = \mathfrak{F}_{SoS}(\Alg)$, and let $\widetilde{R} =
     \langle F, H, I, \widetilde{\Phi}, \widetilde{T}\rangle$ be an
     arbitrary valid run of $\TrAlg$ using $\MFD$.
Let $R^0 = \langle F^0, H, I, \Phi, T\rangle$ be a valid run of $\Alg$
     using $\PFD$ where $\forall t \in \T: F^0(t) = faulty(F)$, $\Phi$
     and $T$ are obtained by deleting the entries associated with
     faulty processes in $\widetilde{\Phi}$ and $\widetilde{T}$,
     respectively.
From Propositions  \ref{prop:RvalidMFDrun}, \ref{prop:R0validMFDrun}
     and \ref{prop:R0validPFDrun}, we know that $R^0$ is a valid run
     of $\Alg$ using $\PFD$.
Therefore, each $w \in \ident{IR}(\Alg,F^0,\PFD,V_\Pi)$  satisfies
     $P(w,F^0)$.

Let $\widetilde{V}_\Pi$ be a function derived from $V_\Pi$ as
     described earlier in this section.
From Propositions \ref{prop:correctProcessSameVi} and
     \ref{prop:faultyProcessSameVi}, we conclude that  for all
     processes $p_i$ and all indexes $\ell$ in runs $\widetilde{R}$
     and $R^0$, $\widetilde{V}_i(\gamma_i(I,\Phi,\ell)) =
     V_i(\gamma_i(I,\Phi,\ell))$.
Note that there is no re-ordering of steps of correct processes
     between $\Phi$ and $\widetilde{\Phi}$; however, steps of faulty
     processes may be missing in $R^0$.
Thus, we infer $\ident{ir}(R^0,V_\Pi) \sqsubseteq
     \ident{ir}(\widetilde{R},V_\Pi)$.
From the finite stuttering property of problems and
     Theorem~\ref{thm:MPcorr}, we conclude that if $\Alg$ solves
     $\problem$ using~$\PFD$ then $ \TrAlg=\mathfrak{F}_{SoS}(A)$
     solves $\problem$ using $\MFD$.
\end{proof}

\begin{corollary}
$\MFD\redSolv\PFD$ and $\PFD\not\redSolv\MFD$.
\end{corollary}

\section{Equivalence among $\PFD_k$ failure detectors}
\label{sec:DK}

\subsection{Algorithmic transformation: Delay-a-Step}
\label{subsubsec:atdas}

Informally, the \emph{Delay-a-Step} transformation (DaS) converts an
     algorithm~$\Alg$ to an algorithm~$\TrAlg$ such that in $\TrAlg$
     each process $p_i$ first executes a single no-op step, and
     subsequently~$p_i$ behaves exactly like it does in $\Alg$.
We define a transformation function $\mathfrak{F}_{DaS}$ that maps an
     algorithm $\Alg = (\Alg_i)_{\forall p_i \in \Pi}$   to a new
     algorithm $\TrAlg = (\TrAlg_i)_{\forall p_i \in \Pi}$.
The new state space of~$\TrAlg$ is constructed as follows.
For each process $p_i$, we add a new set of states $S^\star_i$, which are the initial states of ~$\TrAlg_i$, such
     that $|S^\star_i| = |\hat{Q}_i|$, to obtain the set of states
     for~$\TrAlg_i$.
This implies that the states in $\hat{Q}_i$ are \emph{not} initial
     states of $\TrAlg_i$.
We define a bijective function $\ident{delay}_i\colon S^\star_i
     \rightarrow \hat{Q_i}$.

The state transitions of $\TrAlg$ are the state transition of $\Alg$
     and the following rules for initial states $S^\star_i$: if a
     process $p_i$ is in state $s\in S^\star_i$ when it takes a step,
     then $p_i$ neither receives nor sends messages and goes into
     state $\ident{delay}_i(s)$.

\subsection{Showing $\PFD_{k+1}$ is at least as strong as $\PFD_{k}$}

Let $\Alg$ be an algorithm that solves some problem $\problem$ using a
     failure detector~$\PFD_k$, and let $\TrAlg =
     \mathfrak{F}_{DaS}(A)$.
The remainder of this section  shows that $\TrAlg$ solves $\problem$
     using the failure detector $\PFD_{k+1}$.

Let $\widetilde{R}=\langle \widetilde{F}, \widetilde{H},
     \widetilde{I}, \widetilde{\Phi}, \widetilde{T} \rangle$ be a
     valid run of $\TrAlg$ using $\PFD_{k+1}$.
In the following, we construct (in several steps)  a new initial
     configuration $I$, a new schedule $\Phi$, a new time-sequence
     $T$, a new failure pattern $F$, and a new history $H$ such that
     the run $R = \langle F,H,I,\Phi,T \rangle$ is a valid run of
     $\Alg$ using the failure detector~$\PFD_k$.
We then show that if $R$ is a valid run of $\Alg$ using $\PFD_k$, then
     $\TrAlg$ solves problem $\problem$ using $\PFD_{k+1}$.

First, we construct the initial configuration $I$ as follows.
For each process $p_i$, $I|_i = delay_i(\widetilde{I}|_i)$.

Next, we construct the new schedule $\Phi$ and a new time-sequence
     $T'$ as follows.
For each process $p_i \in \Pi$, let $\ident{no-op}(i)$ denote the
     index of the first entry of the form $(p_i,*,*,*,*,*)$
     in~$\widetilde{\Phi}$.
The schedule $\Phi$ is obtained by deleting for each process $p_i$ the
     step $\widetilde{\Phi}[\ident{no-op}(i)]$
     from~$\widetilde{\Phi}$.
A time-sequence $T'$ is obtained by deleting for each process $p_i$
     the entry $\widetilde{T}[\ident{no-op}(i)]$ from~$\widetilde{T}$.

\begin{proposition}\label{lem:RprimeIsARunOfAlg}
If $\widetilde{R}=\langle \widetilde{F}, \widetilde{H}, \widetilde{I},
     \widetilde{\Phi}, \widetilde{T} \rangle$ is a valid run of
     $\TrAlg$ using $\PFD_{k+1}$ then $R'=\langle \widetilde{F},
     \widetilde{H}, I, \Phi, T' \rangle$ is a valid run of $\Alg$
     using $\PFD_{k+1}$.
\end{proposition}

\begin{proof}
By construction, the first step of each process $p_i$ in $\TrAlg$ is
     of the form $(p_i,*,\bot,d,*,\bot)$, and all the subsequent steps
     of $p_i$ are the same as in $\Alg$.
Since $\widetilde{\Phi}$ is a schedule of $\TrAlg$, we see that for
     each process $p_i$, $\widetilde{\Phi}[\ident{no-op}(i)]$ is the
     first step of $p_i$ executing $\TrAlg_i$, and is therefore a
     no-op step of the form $(p_i,*,\bot,d,*,\bot)$.
Also, note that upon executing a no-op step from
     state~$\widetilde{I}|_i$, process $p_i$ transitions to state
     $delay_i(\widetilde{I}|_i)$ which, by construction, is equal to
     the state~$I|_i$.  

Hence, by deleting the $\widetilde{\Phi}[\ident{no-op}(i)]$ step for each
     process~$p_i$ from~$\widetilde{\Phi}$, we obtain a valid schedule
     for $\Alg$; that is, $\Phi$ is a valid schedule for a run of
     $\Alg$.
Similarly, by deleting the times at which the
     $\widetilde{\Phi}[\ident{no-op}(i)]$ step occurred for each process $p_i$
     from $\widetilde{T}$, we obtain a valid time-sequence for $\Alg$;
     that is, $T'$ is a valid time-sequence for the schedule $\Phi$ in
     a run of $\Alg$.
The proposition follows.
\end{proof}

Then we define the new failure pattern $F$ by $F(t) =
     \widetilde{F}(t+1)$, for $t \in \T$.
Intuitively, each faulty process crashes one time unit earlier in $F$
     than in $\widetilde{F}$.
Similarly, the new history~$H$ is defined by $H(p_i,t) =
     \widetilde{H}(p_i,t+1)$, for all $p_i \in \Pi$ and $t \in \T$.

\begin{proposition}\label{lem:HinDk}
If $\widetilde{H} \in \PFD_{k+1}(\widetilde{F})$  then $H \in \PFD_k(F)$.
\end{proposition} 

\begin{proof} 
Since $\widetilde{H} \in \PFD_{k+1}(\widetilde{F})$, it follows from
$k$-Accuracy that
\begin{equation}\label{eqn:k+1-accuracy}
\forall p_i,p_j \in \Pi,\;
\forall t \in \T\; \colon
 (p_j \notin \widetilde{F}(k+1) \wedge p_j \in 
   \widetilde{H}(p_i, t+1)) \Rightarrow p_j \in \widetilde{F}(t+1),
 \end{equation}
and from $k$-Completeness
\begin{multline}\label{eqn:k+1-completeness}
\forall p_i,p_j \in \Pi,\;
\exists t' \in \T, \;
\forall t>t'\; 
 \colon \\
(p_j \notin \widetilde{F} (k+1) \wedge p_j \notin
 correct(\widetilde{F} ) 
\wedge p_i \in correct (\widetilde{F}) ) \Rightarrow p_j \in \widetilde{H}(p_i, t).
\end{multline}

Since $\forall t \in \T: F(t) = \widetilde{F}(t+1)$, and $\forall p_i
     \in \Pi, \forall t \in \T: H(p_i,t) = \widetilde{H}(p_i,t+1)$,
     substituting these functions in Equations~(\ref{eqn:k+1-accuracy})
     and~(\ref{eqn:k+1-completeness}) we obtain
\begin{equation}\label{eqn:k-accuracy}
\forall p_i,p_j \in \Pi,\;
\forall t \in \T\; :
 (p_j \notin F(k) \wedge p_j \in H(p_i, t)) \Rightarrow p_j \in F(t),
 \end{equation}
and since $correct(F) =
     correct(\widetilde{F})$,
\begin{multline}\label{eqn:k-completeness}
\forall p_i,p_j \in \Pi,\;
\exists t' \in \T, \; 
\forall t>t'\; \colon \\
(p_j \notin F(k) \wedge p_j \notin correct(F) \wedge p_i \in correct
(F))
 \Rightarrow p_j \in H(p_i, t).
\end{multline}
We observe that the failure detector whose histories are as described in
Equations~(\ref{eqn:k-accuracy}) and~(\ref{eqn:k-completeness}) 
     satisfies $k$-Accuracy and $k$-Completeness.
\end{proof}

Because $T'$ is obtained by removing the time of the first step of each
     process, it follows that for any $\ell$, $T'[\ell]>0$.
We may thus define the new time-sequence $T$ as $T[\ell] = T'[\ell]-1$
     with $\ell\in\Naturals$.

\begin{proposition}\label{lem:RIsAValidRun}
If $R'=\langle \widetilde{F}, \widetilde{H}, I, \Phi, T' \rangle$ is a
     valid run of $\Alg$ using $\PFD_{k+1}$, then $R=\langle F, H, I,
     \Phi, T \rangle$ is a valid run of $\Alg$ using $\PFD_{k}$.
\end{proposition}
\begin{proof}
From the construction of $T$, we know that in run $R$, each process
     $p_i$ takes the same steps as in $R'$, but each step taken
     at time $t$ in $R'$ is taken at time $t-1$ in $R$.
From the construction of $H$, we see that the output of the failure
     detector queried in run $R'$ at a time~$t$ is identical to the
     output of the failure detector queried in run $R$ at time $t-1$.
Similarly, in the failure pattern $F$, each process that crashes at time
     $t$ in $\widetilde{F}$ crashes at time~$t-1$ in~$F$.
Therefore, the run $R$ is the run $R'$ after every step and the
     associated failure detector output in $R'$ is moved earlier in
     time by $1$ unit.

Also, recall that  $H \in \PFD_k(F)$, from Proposition~\ref{lem:HinDk}.
Therefore, if $R'$ is a valid run of $\Alg$ using failure detector
     $\PFD_{k+1}$, then $R$ is a valid run of $\Alg$ using $\PFD_{k}$.
\end{proof}

As $\Alg$ solves $\problem$ using $\PFD_k$, for each process $p_i$
     there exists a function $V_i$ that maps each state of $p_i$ to a
     problem state.
For each process $p_i$ we define a new function $\widetilde{V}_i$ as
     follows.
For each (initial) state $s \in S^\star_i$, $\widetilde{V}_i(s) =
     V_i(delay_{i}(s))$, and for each state  $s \notin  S^\star_i$,
     $\widetilde{V}_i(s) = V_i(s)$.

\begin{theorem}
If $\Alg$ solves problem $\problem$ using failure detector $\PFD_{k}$,
     then Algorithm $\TrAlg$ solves problem $\problem$ using failure
     detector $\PFD_{k+1}$.
\end{theorem}
\begin{proof}
Let $\widetilde{R}=\langle \widetilde{F}, \widetilde{H}, \widetilde{I},
     \widetilde{\Phi}, \widetilde{T} \rangle$ be a
     valid, run of $\TrAlg$ using $\PFD_{k+1}$.
Applying Propositions \ref{lem:RprimeIsARunOfAlg}, \ref{lem:HinDk},
     and \ref{lem:RIsAValidRun}, we see that from $\widetilde{R}$ we
     can construct a unique run $R=\langle F, H, I, \Phi, T \rangle$
     that is a valid run of $\Alg$ using $\PFD_{k}$.

Note that by assumption $\Alg$ solves problem $\problem$ using failure
     detector $\PFD_k$.
Hence there is an interpretation $V_\Pi$ which ensures that
     $\problem(\ident{ir}(R,V_\Pi), F)$ holds.
Since $correct(F) = correct(\widetilde{F})$, applying the crash time independence
     property from Section \ref{sec:problemDefinition}, we obtain that
     $\problem(\ident{ir}(R,V_\Pi),\widetilde{F})$ is true.

Note that for each process $p_i$, $p_i$ is never is a state $s_i \in
     S^\star_i$ in run $R$, and for each state  $s \notin  S^\star_i$,
     $\widetilde{V}_i(s) = V_i(s)$.
Therefore, $\ident{ir}(R,V_\Pi) = \ident{ir}(R,\widetilde{V}_\Pi)$.

Also, note that for each process $p_i$, for each state $s \in
     S^\star_i$, $\widetilde{V}_i(s) = V_i(delay_{i}(s))$, and
     $delay_i(s) \in \hat{Q}_i$; therefore,
     $\widetilde{V}_i(delay_{i}(s)) = V_i(delay_{i}(s))$.
In effect, $\ident{ir}(R,\widetilde{V}_\Pi) \sqsubseteq
     \ident{ir}(\widetilde{R}, \widetilde{V}_\Pi)$.
So we apply the finite stutter property from  Section
     \ref{sec:problemDefinition} and see that since
     $\problem(\ident{ir}(R,V_\Pi), F)$ is true,
     $\problem(\ident{ir}(\widetilde{R}, \widetilde{V}_\Pi),
     \widetilde{F})$ is also true.

We thus have shown that for any interpreted run $w \in
     \ident{IR}(\TrAlg,\widetilde{F},\PFD_{k+1},\widetilde{V}_\Pi)$,
     the predicate $P(w,\widetilde{F})$ holds.
In other words,   $\TrAlg$ solves  $\problem$ using
     failure detector $\PFD_{k+1}$.
\end{proof}

\begin{corollary}
$\PFD_k \redSolv \PFD_{k+1}$ and $\PFD_{k+1} \redSolv \PFD_k$.
\end{corollary}

\section{Conclusion}

In this paper, we introduced a new proof technique that compares
     failure detectors and does not depend on the ability of one
     failure detector to implement another.
Instead, we propose a novel approach which is based on algorithm
     transformation  so that for every algorithm $\Alg$ that solves
     some problem using failure detector $\D'$ we derive a new
     algorithm~$\TrAlg$ which solves the same problem using $\D$
     instead, and thus we show $\D \redSolv \D'$, where $\redSolv$ is
     the solvability relation introduced
     in~\cite{charron-bost:10:isolt}.

We demonstrated the utility of the new proof technique by presenting
     two new results.
First, we showed that the $\PFD$ and $\MFD$ failure detectors, which
     are incomparable with respect to the~$\redCT$ and $\redJT$ relations,
     are strictly ordered with respect to the $\redSolv$ relation;
     $\MFD$ is strictly stronger than $\PFD$.
Second, we showed that the $\PFD_k$ series of failure detectors
     (denoted by $\D_k$ in \cite{bhatt:oteow}), which were shown to be
     strictly ordered as $\PFD_k \redCT \PFD_{k+1}$ for all $k$, are
     equivalent to each other with respect to the $\redSolv$ relation.

\paragraph{Significance.}

The primary motivation for the introduction of the $\MFD$ failure
     detector in \cite{guer:01:hfap} was to show that $\PFD$ is not
     the weakest failure detector for certain problems such as
     non-blocking atomic commitment or terminating reliable broadcast.
This was done by showing that $\MFD$ and~$\PFD$, despite being
     incomparable with respect to $\redCT$, can be used to solve the
     aforementioned problems under consideration.
However, we have shown that $\MFD$ and $\PFD$ can be strictly ordered
     with respect to $\redSolv$.
This shows that the reasoning used in \cite{guer:01:hfap} is limited
     only to the  $\redCT$ relation.\footnote{It was later shown in
     \cite{larrea:otwfd} that failure detectors that are weaker with
     respect to $\redCT$ than both $\MFD$ and $\PFD$ are sufficient to
     solve non-blocking atomic commitment and terminating reliable
     broadcast.
However, our motivation was not to find a weakest failure detector for
     a given problem, but rather to make explicit that certain proofs
     are limited to the $\redCT$ relation.}

Similarly, the $\PFD_k$ sequence of failure detectors was introduced
     in \cite{bhatt:oteow} in order to demonstrate that FCFS mutual
     exclusion does not have a weakest failure detector.
The proof relies on the fact that for any $k$, $\PFD_k$ is strictly
     stronger than $\PFD_{k+1}$ with  respect to~$\redCT$
     while every such $\PFD_k$ is sufficient to solve FCFS mutual
     exclusion.
However, we have shown that all the $\PFD_k$ failure detectors are
     equivalent with respect to~$\redSolv$ and, therefore, these
     failure detectors solve the same set of time-free problems.

The above two examples show that some results on weakest failure
     detectors based on the $\redCT$ and $\redJT$ relation do not
     carry over to the $\redSolv$ relation.
This, in conjunction with the seemingly contradictory results
     regarding the (non)existence of weakest failure detectors in
     \cite{jayanti:ephawfd} and \cite{bhatt:oteow}, leaves open the
     possibility that the use of failure detectors as ``computability
     benchmark'' \cite{FGK11} may not be appropriate until we have
     resolved the question of the ``right'' comparison relation to
     order failure detectors.

\paragraph{Comparison to standard proofs.}

From a technical viewpoint, our new proof technique is  quite similar
     to proofs that establish the $\redCT$ relation.
In both, one argues about an algorithm using some failure detector.
In $\redCT$ proofs, one usually gives an algorithm more or less
     explicitly, while we give an algorithm~$\TrAlg$ as function of
     another algorithm~$A$.
In $\redCT$ proofs, one shows that the states the algorithm goes
     through are related to histories of the implemented failure
     detector.
In our $\redSolv$ proofs, we show that the states the algorithm goes
     through are related to problem configuration sequences.

The differences in the comparison relations discussed above then come
     from the fact that we relate to a schedule of algorithm $A$ which
     is within the world of asynchronous runs, while $\redCT$ proofs
     relate to a failure detector history, which is defined with
     respect to time, and is hence outside the world of asynchronous
     runs.

\paragraph{Future Work.} 

Our results are preliminary and provide multiple avenues for future
     work.
We present two such open questions.
First, note that the proof technique introduced here does not
     necessarily characterize the $\redSolv$ relation completely.
That is, there might be other proof techniques which establish the
     $\redSolv$ relation between two failure detectors in the cases
     where our proposed technique does not lead to the required
     result.
Thus, there is scope for complete characterization of the $\redSolv$
     relation.
Second, since different comparison relations establish different
     relationships among various failure detectors, an obvious
     question presents itself: is there a ``right'' comparison
     relation for failure detectors? If yes, which one is it?

\paragraph{Acknowledgement.} We would like to thank Jennifer Welch and Martin Hutle for their comments, suggestions, and criticisms that greatly helped improve this article.

\bibliography{master}
\end{document}